\documentclass[12pt,letterpaper]{article}
\pdfoutput=1
\usepackage{amsmath, amsthm}
\usepackage{amsfonts, amssymb}

\usepackage{setspace}
\usepackage[english]{babel}
\usepackage[utf8]{inputenc}
\usepackage[T1]{fontenc}
\usepackage[tracking]{microtype}
\usepackage[ttscale=0.81]{libertine}
\usepackage[libertine]{newtxmath}
\usepackage{geometry}
\usepackage{fullpage}

\clearpage{}
\usepackage{etoolbox}
\usepackage{cite}
\usepackage{bm}
\usepackage{ifthen}
\usepackage[boxed, linesnumbered]{algorithm2e}

\usepackage{letltxmacro}

\usepackage{url}

\usepackage{nicefrac}
\usepackage{xstring}
\usepackage{pgf}
\usepackage{tikz}
\usetikzlibrary{arrows,automata}

\renewcommand{\mathbf}[1]{\bm{#1}}

\usepackage{graphicx}

\usepackage{booktabs,siunitx}

\usepackage{caption}
\usepackage{subcaption}

\usepackage{thmtools}
\usepackage{thm-restate}

\usepackage{hyperref}
\hypersetup{
  colorlinks=true,
  linkcolor=green!30!black,
  linktoc=all,
  citecolor=blue,
  bookmarksnumbered=true
}
\usepackage{cleveref}

\declaretheorem[name=Theorem]{theorem}

\newtheorem{lemma}[theorem]{Lemma}

\newtheorem{proposition}[theorem]{Proposition}

\theoremstyle{definition}
\newtheorem{definition}[theorem]{Definition}

\AddToHook{env/theorem/begin}{\crefalias{theorem}{theorem}}
\AddToHook{env/conjecture/begin}{\crefalias{theorem}{conjecture}}
\AddToHook{env/lemma/begin}{\crefalias{theorem}{lemma}}
\AddToHook{env/observation/begin}{\crefalias{theorem}{observation}}
\AddToHook{env/claim/begin}{\crefalias{theorem}{claim}}
\AddToHook{env/condition/begin}{\crefalias{theorem}{condition}}
\AddToHook{env/example/begin}{\crefalias{theorem}{example}}
\AddToHook{env/fact/begin}{\crefalias{theorem}{fact}}
\AddToHook{env/corollary/begin}{\crefalias{theorem}{corollary}}
\AddToHook{env/proposition/begin}{\crefalias{theorem}{proposition}}
\AddToHook{env/definition/begin}{\crefalias{theorem}{definition}}
\AddToHook{env/remark/begin}{\crefalias{theorem}{remark}}
\AddToHook{env/problem/begin}{\crefalias{theorem}{problem}}
\AddToHook{env/remark/begin}{\crefalias{theorem}{remark}}

\crefname{theorem}{Theorem}{Theorems}
\crefname{conjecture}{Conjecture}{Conjectures}
\crefname{observation}{Observation}{Observations}
\crefname{claim}{Claim}{Claims}
\crefname{condition}{Condition}{Conditions}
\crefname{example}{Example}{Examples}
\crefname{fact}{Fact}{Facts}
\crefname{lemma}{Lemma}{Lemmas}
\crefname{corollary}{Corollary}{Corollaries}
\crefname{definition}{Definition}{Definitions}
\crefname{remark}{Remark}{Remarks}
\crefname{proposition}{Proposition}{Propositions}
\crefname{question}{Questions}{Questions}
\crefname{problem}{Problem}{Problems}
\crefname{exercise}{Exercise}{Exercises}
\crefname{section}{Section}{Sections}
\crefname{appendix}{Appendix}{Appendices}

\newcommand{\abs}[1]{\ensuremath{\left|#1\right|}}

\newcommand{\ceil}[1]{\ensuremath{\left\lceil#1\right\rceil}}
\newcommand{\floor}[1]{\ensuremath{\left\lfloor#1\right\rfloor}}

\renewcommand{\Pr}[2][]{\ensuremath{{\mathop{\mathrm{Pr}}}_{#1}\insq{#2}}}

\newcommand{\inb}[1]{\left\{#1\right\}}
\newcommand{\inp}[1]{\left(#1\right)}
\newcommand{\insq}[1]{\left[#1\right]}

\makeatletter
\newcommand*{\defeq}{\mathrel{\rlap{\raisebox{0.3ex}{$\m@th\cdot$}}\raisebox{-0.3ex}{$\m@th\cdot$}}=}
\newcommand*{\eqdef}{=
  \mathrel{\rlap{\raisebox{0.3ex}{$\m@th\cdot$}}\raisebox{-0.3ex}{$\m@th\cdot$}}}
\makeatother

\NewDocumentCommand{\blk}{m o}{B_{#1\IfValueT{#2}{,#2}}}

\NewDocumentCommand{\lf}{m}{\mathop{\mathrm{left}}(#1)}

\NewDocumentCommand{\rg}{m}{\mathop{\mathrm{right}}(#1)}

\newcommand{\ham}{\ensuremath{\Delta_{\textup{Hamming}}}}

\newcommand{\im}{\ensuremath{\mathrm{Imm}}}

 \usepackage{color}

\clearpage{}
\newtoggle{arxiv}
\settoggle{arxiv}{true}

\title{
The Rate-Immediacy Barrier in Explicit Tree \\ Code Constructions}

\author{
  Gil Cohen\thanks{
    Department of Computer Science, Tel Aviv University, Tel Aviv, Israel.
    Email: {\texttt gil@tauex.tau.ac.il}. Supported by ERC starting grant 949499.
  }
    \and
  Leonard J. Schulman\thanks{
    E{\&}AS Division, California Institute of Technology, Pasadena, USA.
    Email: {\texttt schulman@caltech.edu}. Supported in part by NSF Award 2321079.
  }
  \and
  Piyush Srivastava\thanks{
    Tata Institute of Fundamental Research, Mumbai, India.
    Email: {\texttt piyush.srivastava@tifr.res.in}.  Supported in part by the Department of
  Atomic Energy, Government of India, under project nos. RTI4001 and RTI4014, by SERB
  MATRICS grant number MTR/2023/001547, and by the
  Infosys-Chandrasekharan virtual center for Random Geometry.
  }
}
\date{}

\clearpage{}
\begin{document}

\maketitle

\iftoggle{arxiv}{\pagenumbering{gobble}}

\begin{abstract}
Since the introduction of tree codes by Schulman (STOC 1993), explicit
construction of asymptotically good tree codes has remained a notorious challenge. A work by Cohen, Haeupler and Schulman (STOC 2018), as well as the state-of-the-art construction by Ben Yaacov, Cohen, and Yankovitz (STOC 2022) have achieved codes with rate $\Omega(1/\log\log n)$, exponentially improving upon the original  rate $\Omega(1/\log n)$ construction of Evans, Klugerman and Schulman from 1994. All of these constructions rely, at least in part, on increasingly sophisticated methods of combining (block) error-correcting codes.

In this work, we identify a fundamental barrier to constructing tree codes using known techniques. We introduce a key property which we call \emph{immediacy}, that, while not required by the original definition of tree codes, is shared by all known constructions and inherently arises in recursive combinations of error-correcting codes. Our main technical contribution is the proof of a \emph{rate-immediacy trade-off}, which, in particular, implies that any tree code with constant distance and non-trivial immediacy must necessarily have vanishing rate. By applying our rate-immediacy trade-off to existing constructions, we establish that their known rate analyses are essentially optimal given their actual error-correction properties. More broadly, our work highlights the need for fundamentally new ideas---beyond the recursive use of error-correcting codes---to achieve substantial progress in explicitly constructing asymptotically good tree codes.

\end{abstract}

\iftoggle{arxiv}{\newpage}{}

 \newcommand{\TagPart}{Tagged partition}
\newcommand{\AtagPart}{A tagged partition}
\newcommand{\TagParts}{Tagged partitions}
\newcommand{\tagPart}{tagged partition}
\newcommand{\atagPart}{a tagged partition}
\newcommand{\tagParts}{tagged partitions}

\newcommand{\LamPart}{Laminar partition}
\newcommand{\AlamPart}{A laminar partition}
\newcommand{\LamParts}{Laminar partitions}
\newcommand{\lamPart}{laminar partition}
\newcommand{\alamPart}{a laminar partition}
\newcommand{\lamParts}{laminar partitions}

\newcommand{\LamCode}{Immediacy code}
\newcommand{\lamCode}{immediacy code}
\newcommand{\LamCodes}{\LamCode{}s}
\newcommand{\lamCodes}{\lamCode{}s}
\newcommand{\AlamCode}{An \lamCode{}}
\newcommand{\alamCode}{an \lamCode{}}

\section{Introduction}
\pagenumbering{arabic} 
\setcounter{page}{1}
Coding theory addresses the problem of communication over an imperfect channel. In the classic setting~\cite{shannon48,hamming1950error}, Alice aims to communicate a message to Bob via a channel that may introduce errors. The central question is: how should Alice encode her message so that Bob can accurately recover it, provided that the number of errors is limited? This scenario motivates the notion of an error-correcting code. Formally, a function $C \colon \Sigma^k \to \Sigma^{n}$ is a \emph{block error-correcting code} with distance $\delta$ and block length \(k\) if, for every pair of distinct strings $x, y \in \Sigma^k$, their respective encodings $C(x)$ and $C(y)$ differ in at least a $\delta$ fraction of positions, with respect to the Hamming distance. The \emph{rate} of information transmission $\rho = \frac{k}{n}$ and the fraction of errors corrected, $\frac{\delta}{2}$, are competing parameters. A fundamental problem in coding theory is to construct explicit \emph{asymptotically good codes}, i.e., codes that maintain constant distance $\delta > 0$ and constant rate $\rho > 0$. Here, by ``explicit'' we mean that the encoding function $C$ can be computed in polynomial time. Justesen~\cite{Justensen} was the first to provide such an explicit construction. Since then, many explicit constructions have been developed (see, e.g.,~\cite{TVZ,Sipser-Spielman}).

While error-correcting codes solve the problem of sending a single message from Alice to Bob, there are scenarios involving dynamic interaction, where messages exchanged depend on previously communicated information. Interactive coding addresses this more intricate problem of enabling reliable interactive communication over imperfect channels. Standard error-correcting codes are insufficient for this task. Readers interested in this rapidly growing research field are encouraged to consult the comprehensive survey by Gelles~\cite{GellesSurvey}. 

Tree codes are crucial combinatorial structures introduced in~\cite{schulman93,schulman96} to facilitate deterministic interactive
coding. Analogous to error-correcting codes in single-message scenarios, tree
codes are trees equipped with a specific distance property. To define this
formally, we first introduce some notation. Let $T$ be an infinite complete
rooted binary tree, with edges labeled from an alphabet $\Sigma$. For two vertices
$u, v$ of equal depth, let their least common ancestor be denoted $w$. Let $\ell$
be the number of edges on the path from $u$ (or $v$) to $w$. Define $p_u,
p_v \in \Sigma^\ell$ as the sequences of symbols labeling the edges on the paths from $w$ to $u$ and from $w$ to $v$, respectively. The quantity $h(u,v)$ denotes the relative Hamming distance between $p_u$ and $p_v$. Intuitively, $h(u,v)$ measures how distinct the sequences labeling the paths to $u$ and $v$ are, disregarding their common prefix. A tree code enforces a minimum bound on this quantity:

\begin{definition}[Tree codes~\cite{schulman93}]
	Let $\delta \in [0,1]$ and let $T$ be an infinite complete rooted binary tree. A labeling of the edges of $T$ from an alphabet $\Sigma$ is called a \emph{tree code with distance $\delta$} if, for every pair of vertices $u, v$ at the same depth, it holds that $h(u, v) \ge \delta$. 
	The \emph{rate} of the tree code, denoted by $\rho$, is defined as $\frac{1}{\log_2 |\Sigma|}$.
\end{definition}

An alternative definition found in the literature describes a tree code as a family \((T_n)_{n \in \mathbb{N}}\), where each \(T_n\) is a rooted binary tree of \emph{finite} depth \(n\). Such a family is called a tree code with distance \(\delta\) if each \(T_n\) has a distance of at least \(\delta\), as previously defined. Clearly, an infinite tree code naturally induces such a finite family when truncated at any given depth. Conversely, it has been shown that the opposite direction also holds~\cite{benyaacovCY22}, with only a constant degradation in parameters. Hence, we use these two definitions interchangeably throughout this informal discussion.

Initially, it was not clear whether an asymptotically good tree code—one with
both positive rate and positive distance—existed. Schulman provided three
distinct proofs showing that, for any constant $\delta < 1$, a tree code exists with
alphabet size $|\Sigma| = O_\delta(1)$ that achieves distance $\delta$. More recently, by
adapting one of these approaches, it was proved that there is a tree code with
just $|\Sigma| = 4$ symbols, namely a rate-$\frac1{2}$ tree code, and positive
distance (specifically, $\delta > 0.136$)~\cite{CS20}.  It was also observed in
\cite{CS20} that $3$ symbols are insufficient for guaranteeing positive distance.
However, all known existence proofs of asymptotically good tree codes are
non-explicit, relying on probabilistic methods in non-trivial ways. Despite significant interest, the explicit construction of asymptotically good tree codes remains a notorious open problem.

Given these difficulties, researchers have naturally considered constructing
tree codes allowing for vanishing rate, where the objective is nonetheless to
minimize the rate deterioration as a function of the depth $n$.\footnote{We
  sometimes also refer to the depth \(n\) of a tree code as its
  \emph{transmission length}.} A trivial construction which encodes the entire
path from the root on each edge, achieves $\delta = 1$ but rate $\frac1{n}$. In an
unpublished manuscript, Evans, Klugerman, and Schulman~\cite{EvansKS94} provided
a construction with rate $\Omega(1/\log n)$. An improvement was made only fairly
recently by Cohen, Haeupler and Schulman \cite{CohenHS18}, who constructed tree
codes with rate $\Omega(1/\log\log n)$. A decoding algorithm to the latter tree code
construction was devised by Narayanan and Weidner~\cite{NarayananW20}, who also
suggested alternative constructions. Connections between~\cite{CohenHS18} and
the work of Pudl{\'a}k \cite{pudlak2013linear} (see also \cite{Yanko25} for a
related construction) were explored by Bhandari and
Harsha~\cite{BH20}. Additionally, two distinct explicit constructions with
constant rate were proposed by Moore and Schulman~\cite{moore2014tree} and by
Ben Yaacov, Cohen, and Narayanan \cite{BenYaacovCN21}, but their correctness
hinges upon plausible but unproven conjectures.

The state-of-the-art construction by Ben Yaacov, Cohen, and Yankovitz \cite{benyaacovCY22} also achieves a rate of \(\Omega(1/\log\log n)\) for constant distance \(\delta\), but it improves upon the dependence on the distance parameter. Specifically, while the construction of \cite{CohenHS18} has a rate that is upper bounded by \(O(1/\log\log n)\) regardless of the value of \(\delta\), the construction of \cite{benyaacovCY22} achieves rate  \(\Omega(1/(\sqrt{\delta} \cdot \log\log n))\). In particular, it can attain \emph{constant} rate by compromising on the distance, achieving a distance of \(\delta = \Omega(1/(\log\log n)^2)\). This related, “dual” problem of constructing a constant-rate tree code with slowly deteriorating distance was first studied by Gelles, Haeupler, Kol, Ron-Zewi, and Wigderson \cite{GHKRZW15}, who showed how to achieve distance \(\Omega(1/\log n)\).

\subsection{Immediacy}

Given that we still have no explicit construction of asymptotically good tree codes—despite significant efforts for over three decades—it is important to better understand why current techniques fall short. Identifying a barrier may help steer research away from potential dead ends. More specifically, since tree codes induce standard block error-correcting codes, while the latter are heavily relied upon in all known constructions (except for the two conjecture-based approaches \cite{moore2014tree, BenYaacovCN21}), it is natural to ask: to what extent might constant-rate tree codes be constructed from block codes? Exploring this question may shed light on the fundamental nature of tree codes and clarify how ``close'' they truly are to their well-understood counterparts, block error-correcting codes.

The main contribution of this paper is identifying a fundamental barrier in the
construction of tree codes. We observe that each of the existing explicit
constructions of tree codes (except those relying on unproven conjectures) satisfies a stronger guarantee than strictly required for a tree code; we refer to this guarantee as \emph{immediacy}. This stronger property arose inherently due to the recursive nature of existing constructions \cite{EvansKS94,CohenHS18,NarayananW20,benyaacovCY22}. 
We prove that the immediacy property incurs a cost in rate. Specifically, we
establish a quantitative \emph{rate-immediacy trade-off}, precisely capturing
the achievable rate given a code's immediacy. Our results show that existing
tree code constructions essentially achieve an optimal rate-immediacy balance,
thus clarifying why their rate inevitably vanishes. More broadly, our work
indicates that the typically employed framework, that of tiling standard
error-correcting block codes in progressively sophisticated ways, cannot yield asymptotically good tree codes. Such strategies inherently satisfy the immediacy property, constraining their potential rate.

We formally define immediacy in detail later, as its precise definition is
somewhat technical, involving a certain set system (see the definition of
\emph{\lamCodes{}} in Section \ref{sec:laminar-codes}). However, to capture the
essence, we first provide an informal, simplified definition in which immediacy
is represented by a monotone-increasing function
\(\im : \mathbb{N} \to \mathbb{N}\). A tree code \(C\) is said to have immediacy
\(\im\) if, for every pair of messages \(x, y \in \Sigma^n\) (where \(n\) is the
transmission length), every index \(s\) where
\(x_s \neq y_s\), and each integer \(k \geq 1\), the relative Hamming distance between
\(C(x)\) and \(C(y)\), restricted to the interval \([s, s + \im(k))\) (provided
that \(s + \im(k)\) is no larger than the depth of the tree, if finite), is
bounded below by \(\delta\).  

The definition of a standard tree code imposes no non-trivial immediacy constraint. Disagreements between corresponding codewords are required to accumulate only starting from the \emph{first} index \( s \) at which the two messages \( x \) and \( y \) differ.
In contrast, a code admitting an immediacy function requires that disagreements accumulate starting from \emph{every} index \(s\) for which \(x_s \neq y_s\), at interval lengths determined by the immediacy function \(\im\). Put differently, unlike a standard tree code, a tree code with immediacy function \(\im\) supports \emph{cold-start decoding}. That is,
if during any interval \([s, s + \im(k))\) the noise rate is below $\delta/2$, then the receiver can decode the \(s\)-th bit---even if the noise prior to time $s$ was high and bits prior to $s$ cannot yet be reconstructed.
While this property may be advantageous in certain applications (and is typical of weakly-encoded communications such as human speech or packetized transmissions), our results show that it is an ``expensive'' feature that cannot be afforded when aiming to construct asymptotically good tree codes.

To take an example, it can be shown by inspection that the first tree code
construction \cite{EvansKS94}, which has a rate of \(\Omega(1/\log n)\) (where
\(n\) is the depth of the tree), {possesses a non-trivial immediacy property
  that corresponds roughly to} \(\im(k) = 2^{\Theta(k)}\), which we denote throughout
{this informal presentation} as $\exp(k)$. At this informal
stage, it is convenient to regard \(\im\) as a real-valued function—which in
this case is \(\im(x) = \exp(x)\)—thus extending the original discrete
definition. With this, our main result can be informally stated as follows:

\begin{theorem}[Main result; informal]\label{thm:main-informal}
  Let \(T\) be a depth-\(n\) tree code with constant distance and immediacy
  function \(\im\).  Then, subject to some technical conditions\footnote{For
    instance, it turns out that exponential immediacy {(i.e.
      \(\im(k) = \exp(k)\) as described above)} implies \emph{linear}
    immediacy. Consequently, Theorem \ref{thm:main-informal} should be
    interpreted as applying only to exponential (or faster-growing) immediacy.
  } on \(\im\), that are satisfied for the cases \( \im(k) = \exp(k)\) and
  \(\im(k) = \exp(\exp(k) )\), the rate of \(T\) satisfies
  \[
    \rho = O\left(\frac{1}{\im^{-1}(n)}\right).
  \]
\end{theorem}

Recall that for the construction given in \cite{EvansKS94}, the immediacy
corresponds roughly to $\im(x) = \exp(x)$ and so
\(\im^{-1}(x) = \Theta(\log x)\). Thus, Theorem \ref{thm:main-informal} indicates that
the rate of such a construction cannot exceed \(O(1/\log n)\)---which matches its
known rate, up to the constant.
Theorem \ref{thm:main-informal} also explains the rate achieved by the state-of-the-art tree code
constructions \cite{CohenHS18,NarayananW20,benyaacovCY22}. Upon inspection, these
constructions—similarly to \cite{EvansKS94}—also rely on composition of
error-correcting codes and possess a non-trivial, though exponentially weaker,
immediacy that corresponds roughly to \(\im(x) = \exp(\exp(x))\).
Thus, Theorem \ref{thm:main-informal} indicates that the rate of these
constructions cannot exceed \(O(1/\log \log n)\), matching the proven rate of
these codes.

\paragraph*{Weaker immediacy properties} We emphasize here, however, that the
precise immediacy properties enjoyed by the \cite{EvansKS94} and
\cite{CohenHS18} constructions are slightly weaker than that implied by an
immediacy function, and so \cref{thm:main-informal} does not \emph{formally} apply to
them. Roughly speaking, these constructions only admit cold-start decoding after
a ``lag''.  We show in \cref{sec:appl-rate-upper-bounds} that our general
framework of \emph{\lamCodes} and our general rate upper bound results in
\cref{thm-lamcode-rate-upper-bound,thm-deficient-lamcode-rate-upper-bound}
discussed in \cref{sec-main-results} are powerful enough to extend to
constructions with such weaker forms of immediacy.  Thus, as shown in
\cref{sec:appl-rate-upper-bounds}, the rate upper bounds indicated by
\cref{thm:main-informal} indeed apply to the above constructions. (Further, as we show in \cref{sec:immed-funct-lamc}, constructions satisfying
the simplified immediacy condition described here are also immediacy codes with
appropriate parameters and our general rate upper bounds proved in
\cref{thm-lamcode-rate-upper-bound,thm-deficient-lamcode-rate-upper-bound}
therefore also apply to them.)

More generally, immediacy should be kept in mind when constructing tree codes:
immediacy properties of a proposed construction can be analyzed in order to
obtain an upper bound on the rate it can achieve. (Further, the above discussion
also shows that a rough analysis along the lines of \cref{thm:main-informal} may
serve as a useful guide before a formal analysis of the exact immediacy
properties of a proposed construction.)
It is worth noting that the two constructions whose analyses rely on unproven
conjectures \cite{moore2014tree, BenYaacovCN21} do not appear to exhibit
non-trivial immediacy. Therefore, they are not excluded from being
asymptotically good codes by this work.

  \paragraph*{Rate–distance–immediacy trade-off}
  We point out that our rate upper bounds apply also to codes with vanishing
  distance. In fact, for a general relative distance $\delta > 0$, the upper bound on
  the rate $\rho$ stated in Theorem \ref{thm:main-informal} exhibits, roughly, an
  inverse-linear dependence on $\delta$. The precise details are more nuanced, and we
  direct the reader to the formal statement of our main result in Theorem
  \ref{thm-lamcode-rate-upper-bound}, and to the computations leading to
  \cref{eq:27,eq:22} (that show that fast growing immediacy functions imply an
  appropriate immediacy codes condition).  The key takeaway, however, is that
  Theorem \ref{thm:main-informal}, in its full generality, yields a
  \emph{rate–distance–immediacy trade-off}. For ease of presentation, we choose
  to focus on the trade-off between rate and immediacy in this introductory
  section (see Section \ref{sec:overview} for further discussion).

\subsection{Related work}

\paragraph*{Randomized vs.\ deterministic encoding schemes} Randomized hashing
schemes were employed for online and interactive communication~\cite{schulman92}
before deterministic schemes were made possible by tree codes.  Randomized
schemes do not suffer a trade-off between rate and immediacy. Computationally
efficient encoding and decoding was shown for such schemes
in~\cite{BK12,BN13,GH14}; in a considerable accomplishment Brakerski, Kalai, and
Saxena \cite{brakerskiKS20} later showed that, with access only
to explicit (i.e., polynomial-time encodable) tree codes, an interactive communication
can be performed at the rate of those tree codes, while also performing all
necessary decoding in polynomial time.

\paragraph*{Tree code variants} The original motivation for tree codes
  came from the question of developing analogues of Shannon's channel coding
  theorem (for one-way communication) for situations involving interaction.
  Using the existence of tree codes, such an interactive coding theorem was
  proved by Schulman~\cite[Theorem 2]{schulman96}.  Later work by Gelles,
  Moitra, and Sahai \cite{GMS11} showed that such an interactive coding theorem
  can also be established using a weaker variant of tree codes that they called
  \emph{potent tree codes}. 
On the other hand, stronger variants of tree codes featuring local testability
were recently introduced by Moud, Rosen, and Rothblum \cite{MRR25}. It is
interesting to compare this local testability guarantee with the immediacy
property of tree codes, as immediacy also provides a form of local guarantee,
albeit of a seemingly different nature. Lastly, we mention that a
signal-processing analogue of tree codes was introduced by Schulman and
Srivastava\cite{SchulmanS19}, and a variant of tree codes, dubbed
\emph{palette-alternating tree codes} which allows one to bypass the
$\frac1{2}$-rate barrier was introduced by Cohen and Samocha \cite{CS20}.

\subsection{Technical overview and organization of the paper}  
\label{sec:overview}

We follow the tradition of using information theoretic tools for bounding
capabilities of codes.  Information theory had its genesis in understanding
fundamental limitations on source and channel coding in the stochastic error
model.  However, it has also been useful in understanding limitations on codes in
the ``Hamming model'' of a bounded number of worst case errors, including in the
setting of more modern coding-theoretic notions such as the notion of local
decoding analyzed by Katz and Trevisan~\cite[Section
3.1]{katzEfficiencyLocalDecoding2000}, as well as its more delicate relaxed
variant analyzed by Gur and Lachish~\cite{GurL21} (see also
\cite{DallGL23,goldreich2023lower}).  In our application to understanding
immediacy properties of tree codes, the crucial ingredient turns out to be a
careful accounting of the ``common information'' between different parts of the
codeword, performed simultaneously at different length scales.

{
\interfootnotelinepenalty=10000

The main conceptual contribution of this paper, introduced in
\cref{sec:laminar-codes}, is the notion of \emph{\lamCodes{}}. The notion of
\lamCodes{} provides a more robust formulation of the informal notion of
immediacy functions discussed in the introduction (the close connection between
the two notions is further explored in \cref{sec:immed-funct-lamc}), and also
allows us to smoothly carry out the accounting of common information alluded to
above. The main information theoretic tool we need for this accounting is a form
of the data processing inequality (\cref{lem-data-processing})\footnote{This
  corollary of the (proof of the) data processing inequality also underlies the
  formulation of the \emph{Gács-Körner common information}~(see, e.g.,
  \cite[Section III]{kamathNewDualGacsKorner2010}), but we do not need this
  connection for this paper.} to quantify ``common information'' between
different parts of a code word. Our main technical result is that
\emph{\lamCodes{} cannot have high rates} (\cref{thm-lamcode-rate-upper-bound}):
this is proved in \cref{sec:rate-upper-bound}. The principle underlying the
definition of \lamCodes{} and the proof of \cref{thm-lamcode-rate-upper-bound}
is a careful accounting of the amount of ``information'' that needs to be
duplicated in order to provide the immediacy guarantees that the code provides.
We emphasize that once we have set up the definition of \lamCodes{} and the
requisite form of the data processing inequality, the proof of
\cref{thm-lamcode-rate-upper-bound} is technically quite simple.

}

Finally, in \cref{sec:appl-rate-upper-bounds}, we apply our framework to study
the immediacy properties of the tree code constructions proposed in
\cite{EvansKS94}, \cite{CohenHS18} and \cite{GHKRZW15}. As discussed earlier,
the immediacy properties of these constructions differ slightly from the
informal notion introduced above. For the first two constructions, we show that
the rates achieved by these codes are optimal—up to constant factors—with
respect to the immediacy guarantees they provide. In
  \cref{sec-subexp-time-constr}, we consider a method due to
  Braverman~\cite{Braverman12} for constructing depth-\(n\) good tree codes in
  time sub-exponential in \(n\), and show that the rate-runtime trade-off
  obtained by the best currently known version of this construction is
  essentially the best possible given the immediacy guarantees achieved by the
  construction.

The construction of \cite{GHKRZW15} attempts to tackle a different trade-off: it
achieves an \(\Omega(1/\log n)\) distance with constant rate.  We show that for the
immediacy guarantee that this construction provides, no constant rate tree code
can improve the distance to \(\omega\inp{\frac{\log \log n}{\log n}}\).  This
establishes that the rate/distance trade-off achieved by this construction is
close to the best possible given the immediacy guarantees it achieves.

We do not carry out a full analysis of the immediacy properties of the constructions in \cite{NarayananW20, benyaacovCY22}. The construction of \cite{NarayananW20} employs the same ``code tiling'' structure as \cite{CohenHS18}; the main difference lies in the ``outer'' large-alphabet tree code used. Since immediacy is determined primarily by the code tiling aspect, the analysis of \cite{NarayananW20} proceeds similarly to our analysis of \cite{CohenHS18}. The construction of \cite{benyaacovCY22}, in contrast, uses a somewhat different code-tiling strategy. Nevertheless, a variant of our analysis for \cite{CohenHS18} can still be adapted to study its immediacy properties. Lastly, we note that the construction of \cite{GHKRZW15} is similar to that of \cite{EvansKS94} with respect to the code-tiling component. As such, our result—when viewed not as a rate–immediacy trade-off but rather as a distance–immediacy trade-off—can be used to establish an upper bound on the distance achievable by the construction.

\medskip

We begin in the next section with some technical preliminaries.

\section{Preliminaries}
\label{sec:preliminaries}

\paragraph*{Notation} We will view strings \(s\) of length \(\ell\) as indexed by
integers in \(\inb{1, 2, \dots, \ell} \eqdef [\ell]\). Given such a string, and a
subset \(A \subseteq [\ell]\) we denote by
\(s_A\) the sub-sequence of \(s\) obtained by taking the letters at the indices in
\(A\).  We will sometimes overload the usual notation for real intervals (e.g.
\([a, b], (a, b] \) etc.) to refer only to the \emph{integers} in those
intervals: this overloading should always be clear from the context. We denote
\(\log_2 x\) by \(\lg x\) for any positive real \(x\).  

\subsection{Information theory}
We use standard information theoretic notation for entropy, conditional entropy
and mutual information~(see, e.g., \cite{coverT06}), which we now proceed to
review.  Given a random variable \(X\) taking values in some finite set
\(\Omega\), its \emph{entropy} \(H(X)\) is defined as
\(-\sum_{\omega \in \Omega}\Pr{X = \omega}\cdot \lg \Pr{X = \omega}\) (so we measure entropy in bits), 
with the convention $0 \cdot \lg 0 = 0$.
Given two random variables \(X\) and \(Y\), both taking values
in (possibly different) finite sets, the \emph{conditional entropy} \(H(X|Y)\)
is defined as
\begin{equation}
  \label{eq:18}
  H(X|Y) \defeq -\sum_{x, y}\Pr{X = x, Y = y}\cdot \lg \Pr{X = x | Y = y}.
\end{equation}
The \emph{mutual information} \(I(X : Y)\) is then defined as
\begin{equation}
  \label{eq:20}
  I(X : Y) = H(X) - H(X|Y).
\end{equation}
If \(Z\) is another random variable taking values in a finite set, the
\emph{conditional mutual information} \(I(X : Y | Z)\) is defined as
\begin{equation}
  \label{eq:21}
  I(X : Y | Z ) = H(X|Z) - H(X|Y, Z).
\end{equation}
We collect here some of the standard properties of these quantities, which we shall subsequently use without comment. 
Proofs of these properties can be found, e.g., in~\cite{coverT06}.
\begin{proposition}[\textbf{Properties of mutual information and entropy}]
  Let \(X, Y\) and \(Z\) be random variables taking values in finite sets
  \(\Omega_1, \Omega_2, \Omega_3\) respectively.  Then
\begin{itemize}
\item (Non-negativity of the entropy) \(0 \leq H(X) \leq {\lg \abs{\Omega_1}}\).
\item (Non-negativity of the mutual information) \(I(X : Y) \geq 0\) and
  \(I(X:Y | Z ) \geq 0\).
\item (Conditioning reduces entropy) \(H(X|Y, Z) \leq H(X|Y) \leq H(X)\).
\item (Chain rule for entropy) \(H(X, Y) = H(X) + H(Y|X) = H(Y) + H(X|Y).\)
\item (Sub-additivity of entropy) \(H(X, Y) \leq H(X) + H(Y)\).
\item (Chain rule for mutual information)
  \(I(X, Y : Z) = I(X : Z) + I(Y : Z | X) = I(Y : Z) + I(X: Z | Y)\).
\item (Conditional entropy of deterministic functions is zero) If there is a
  function \(f: \Omega_2 \rightarrow \Omega_1\) such that \(f(Y) = X\), then \(H(X | Y) = 0\).
\end{itemize}
\end{proposition}

\section{\LamCodes{}}
\label{sec:laminar-codes}

In this section, we first introduce the notion of \emph{\lamCodes{}}
(~\cref{def-lamcodes}), a formal abstraction of the informal notion of immediacy
discussed in the introduction. To connect the two notions, we show in
\cref{sec:immed-funct-lamc} how codes with the two kinds of immediacy functions
discussed in the introduction satisfy the \lamCode{} property.  In
\cref{sec:rate-upper-bound}, we prove our main technical result
(\cref{thm-lamcode-rate-upper-bound}): a general rate upper bound for
\lamCodes{}.

The motivation behind our notion of \lamCodes{} is to capture the
\emph{cold-start decoding} property implied by immediacy: if the noise in a
given interval is small, then a code with immediacy properties should be able to
decode at least part of the interval.  For greater applicability of our
framework, we would not like to impose such a condition on \emph{all} intervals, and
towards this goal we begin with a few combinatorial definitions.  By a
\emph{partition} of a set \(U\), we mean an ordered tuple of pairwise disjoint
non-empty subsets of \(U\) whose union is \(U\).

\begin{definition}[\textbf{\TagParts{}}]
  A \emph{\tagPart{}} \(S\) of \([n]\) is a partition \((S_1, S_2, \dots, S_k)\)
  of \([n]\), along with a partition \((\lf{S_i}, \rg{S_i})\) of each
  \(S_i\), \(1 \leq i \leq k\), into two non-empty sets.
\end{definition}

\begin{definition}[\textbf{\LamParts{}}]
  An \emph{\((\alpha, \ell)\)-\lamPart{}} of \(n\) is a tuple
  \(P = (P_0, P_1, P_2, \dots, P_{\ell})\) of partitions of \([n]\) into disjoint
  non-empty subsets, in which \(P_1, \dots, P_{\ell}\) are \tagParts{}, and
  which satisfies the following properties:
  \begin{enumerate}
  \item \textbf{(Size property)} For each \(1 \leq i \leq \ell\), and for each
    \(B \in P_i\), \(\abs{\lf{B}} \geq \alpha\abs{B}\). (Here,
    \(\lf{\cdot}\) is as in the definition of \atagPart{}).
  \item \textbf{(The laminar property)} For each partition \(P_i\) with
    \(i \geq 1\), and any subset \(B \in P_i\), there {exist sub-collections} of
    \(P_{i-1}\) that are partitions of {\(\lf{B}\) and \(\rg{B}\),
      respectively,} into pairwise disjoint non-empty sets (see
    \cref{fig-lamparts} for an illustration).
  \end{enumerate}

\end{definition}

\begin{figure}[h]
  \centering

  \begin{tikzpicture}
\newlength{\rlen}
    \newlength{\rheight}
    \newlength{\rgap}
    \pgfmathsetlength{\rlen}{3cm}   \pgfmathsetlength{\rheight}{1cm}   \pgfmathsetlength{\rgap}{0.3cm}  

\draw[thick] (-\rlen, 0) -- (-\rlen, \rheight) -- (3*\rlen, \rheight) -- (3*\rlen, 0) -- cycle;
    \node at (-\rlen/2, 0.5*\rheight) {\(\lf{B}\)};
    \node at (1.5*\rlen, 0.5*\rheight) {\(\rg{B}\)};

\foreach \i in {-2,-1,...,5} {
        \draw[thick] (\i*\rlen/2, -\rheight-\rgap) rectangle ({(\i+1)*\rlen/2}, -\rgap);
    }

\foreach \x in {-1, 0, 3} {
        \draw[dashed, opacity=0.7] (\x*\rlen, 1.4*\rheight) -- (\x*\rlen,
        -1.4*\rheight -\rgap);
    }

\node at (-1.3*\rlen, 0.5*\rheight) {\large \(\cdots\)};

    \node[label=right:{\large \(P_{i}\)}] at (3.3*\rlen, 0.5*\rheight) {\large \(\cdots\)};

    \node at (-1.3*\rlen, -0.5*\rheight-\rgap) {\large \(\cdots\)};

    \node[label=right:{\large \(P_{i-1}\)}] at (3.3*\rlen, -0.5*\rheight-\rgap)  {\large \(\cdots\)};;

  \end{tikzpicture}
  \caption{An illustration of \lamParts{}}
  \label{fig-lamparts}
\end{figure}

\begin{definition}[\textbf{\LamCodes{}}]\label{def-lamcodes}
  An \emph{\((\alpha, \ell)\)-\lamCode{}} with code alphabet \(\Sigma\) and transmission
  length \(n\) is a function \(c: \Sigma_{\mathrm{in}}^n \rightarrow \Sigma^n\) along with an
  \((\alpha, \ell)\)-\lamPart{} \(P = (P_0, P_1, P_2, \dots, P_{\ell})\) of
  \([n]\) such that
 \begin{itemize}
 \item \textbf{(The code decodes a large neighborhood)} If \(B\) is a subset in
   a partition \(P_i \in P\) for \(i \geq 1\), then for all
   \(x, y \in \Sigma_{\mathrm{in}}^n\), \(x_{\lf{B}} \neq y_{\lf{B}}\) implies that
   \(c(x)_{\rg{B}} \neq c(y)_{\rg{B}}\).  In other words, there is a function
   \(\phi_B: \Sigma^{|\rg{B}|} \rightarrow \Sigma_{\mathrm{in}}^{|\lf{B}|}\) such that
   \(\phi_B(c(x)_{\rg{B}}) = x_{\lf{B}}\) for all \(x \in \Sigma_{\mathrm{in}}^n.\) \end{itemize}
\end{definition}

For some applications, the neighborhood decoding condition above may not hold
for all subsets in all partitions.  For handling these cases, we will also need
the following weakening of the above definition.

\begin{definition}[\textbf{Deficient \LamCode{}}]\label{def-deficient-lamcodes}
  A \(D\)-deficient \((\alpha, \ell)\)-\lamCode{} with code alphabet \(\Sigma\) and transmission length
  \(n\) is a function \(c: \Sigma_{\mathrm{in}}^n \rightarrow \Sigma^n\) along with an
  \((\alpha, \ell)\)-\lamPart{} \(P = (P_0, P_1, P_2, \dots, P_{\ell})\) of
  \([n]\) such that
 \begin{enumerate}
 \item There exist \emph{deficiency subsets} \(S_i\) of \(P_i\), for
   \(1 \leq i \leq \ell\), such that the total size of all the subsets contained in the
   \(S_i\) is at most \(D\).  In symbols,
   \begin{equation}
     \label{eq:28}
     \sum_{i = 1}^{\ell}\sum_{B \in S_i} \abs{B} \leq D.
   \end{equation}
 \item \textbf{(The code often decodes a large neighborhood)} If \(B\) is a
   subset in a partition \(P_i \in P\) for \(i \geq 1\) such that \(B\) is not an
   element of the corresponding deficiency set \(S_i\), then for all
   \(x, y \in \Sigma_{\mathrm{in}}^n\), \(x_{\lf{B}} \neq y_{\lf{B}}\) implies that
   \(c(x)_{\rg{B}} \neq c(y)_{\rg{B}}\).  In other words, there is a function
   \(\phi_B: \Sigma^{|\rg{B}|} \rightarrow \Sigma_{\mathrm{in}}^{|\lf{B}|}\) such that
   \(\phi_B(c(x)_{\rg{B}}) = x_{\lf{B}}\) for all \(x \in \Sigma_{\mathrm{in}}^n.\) \end{enumerate}
\end{definition}

\subsection{Rate upper bounds for \lamCodes{}}
\label{sec-main-results}
We can now state our main technical results, which give rate upper bounds for
(deficient) \lamCodes{}.

\begin{restatable}[\textbf{Rate upper bound for \lamCodes{}}]{theorem}{lamUpperBound}\label{thm-lamcode-rate-upper-bound}
  Let \(n\) be a positive integer, and suppose that
  \(c: \Sigma_{\mathrm{in}}^n \rightarrow \Sigma^n\) is an \((\alpha(n), \ell(n))\)-\lamCode{}. Then,
  \begin{equation}
    \label{eq:19}
    \lg{\abs{\Sigma}} \geq \alpha(n) \ell(n)  \lg \abs{\Sigma_{\mathrm{in}}}.
  \end{equation}
\end{restatable}

\begin{restatable}[\textbf{Rate upper bound for deficient \lamCodes{}}]{theorem}{defLamUpperBound}\label{thm-deficient-lamcode-rate-upper-bound}
  Let \(n\) be a positive integer, and suppose that
  \(c: \Sigma_{\mathrm{in}}^n \rightarrow \Sigma^n\) is a \(D\)-deficient
  \((\alpha(n), \ell(n))\)-\lamCode{}. Then,
  \begin{equation}
    \lg{\abs{\Sigma}} \geq \alpha(n) \cdot \inp{\ell(n) - \frac{D}{n}}\cdot \lg \abs{\Sigma_{\mathrm{in}}}.
  \end{equation}
\end{restatable}

\paragraph*{Remark} In \cref{sec:immed-funct-lamc}, we describe how the notion
of immediacy handled in the above theorems relates to the informal discussion in
the Introduction in terms of immediacy functions.  We then show how the informal
statement of \cref{thm:main-informal} can be formalized using the above
theorems.  Note, however, that the more abstract nature of the above theorems
allows us to handle the subtler immediacy properties of the known tree code
constructions we consider in \cref{sec:appl-rate-upper-bounds}.

\subsection{Immediacy functions and \lamCodes{}}
\label{sec:immed-funct-lamc}

We now show that tree codes satisfying the immediacy conditions described in the
introduction also satisfy the \lamCode{} condition, with parameters depending
upon the distance and the immediacy function.  In order to avoid technical
issues such as divisibility, we make here some simplifying assumptions on the
form of the immediacy function. We note that these assumptions apply to the
examples of immediacy functions discussed in the introduction, and also that the
framework of \lamCodes{} is flexible enough that one can work through a similar
route even when they do not hold.

Let \(T: \inb{0, 1}^n \rightarrow \Sigma^n\) be a tree code with immediacy function
\(\im\), distance parameter \(\delta \in (0, 1)\), and suppose that its depth
\(n = 2\cdot\im(\ell t)\), where \(\ell\) and \(t\) are positive integers, with \(t = t(\delta)\)
possibly depending upon the distance \(\delta\).  Let
\(\kappa \defeq \floor{\lg (2/\delta)}\) be a positive integer so that
\(2^{-\kappa} < \delta \leq 2^{-\kappa + 1}\).  We assume that \(t\) can be chosen so that for
each \(1 \leq j \leq \ell\),
\begin{equation}
  \label{eq:24}
  \begin{gathered}
    2^{-\kappa} \cdot \im(j t) \text{ is an integer, and}\\
    2 \cdot \im((j-1)t)  \text{ divides } 2^{-\kappa} \cdot \im(j t)
  \end{gathered}
\end{equation}
(We note that this is the main technical assumption that would need to be
modified when working with immediacy functions other than those discussed in the
introduction. For the two immediacy functions discussed in the introduction, we
show how to choose such a \(t\) towards the end of this section.)

Given these preliminaries, we now show that \(T\) is an
\((\alpha, \ell)\)-\lamCode{} with \(\alpha \defeq 2^{-(\kappa + 1)}\) and
\(\ell\) as above.  We first show that with the divisibility condition in
\cref{eq:24}, we naturally obtain an \((\alpha, \ell)\) \lamPart{}
\((P_0, P_1, P_2, \dots, P_{\ell})\).  For each \(0 \leq j \leq \ell\), partition
\([n] = [2\im(\ell t)]\) into consecutive blocks of length \(2\im(j t)\) each, and
let these blocks constitute the \tagPart{} \(P_j\).  Given a block \(B \in P_j\)
for \(j \geq 1\), we define \(\lf{B}\) to be the set of the leftmost
\(2^{-\kappa} \cdot \im(jt)\) integers in \(B\) and \(\rg{B}\) to be the rightmost
\(2\im(j t) - \abs{\lf{B}}\) integers in \(B\). 
Given the divisibility
conditions in \cref{eq:24}, we can then verify by a direct computation that both
the laminar and size properties in the definition of an
\((\alpha, \ell)\)-\lamPart{} are satisfied by this construction.  {In
  particular, \(\lf{B}\) and \(\rg{B}\) are disjoint unions of blocks in
  \(P_{j-1}\).}

We now show that \(T\) is also \alamCode{} with respect to the tagged partition
\((P_0, P_1, \dots, P_\ell)\).  Fix any \(1 \leq j \leq \ell\) and any block
\(B\) in the \tagPart{} \(P_j\).  We only need to show that if
\(x, y \in \inb{0, 1}^n\) differ on some index in \(\lf{B}\), then \(T(x)\) and
\(T(y)\) differ on some index in \(\rg{B}\).  For any such \(x, y\), let \(i\)
be the smallest index in \(B\) on which \(x\) and \(y\) differ, and let \(S\) be
the interval \([i, i + \im(j t))\) of length \(\im(j t)\).  By the choice of
\(x \) and \(y\), we have \(i \in \lf{B}\).  Since
\(\abs{\lf{B}} = \im(j t) \cdot 2^{-\kappa}\), \(\abs{B} = 2 \im (j t)\), and
\(\kappa \geq 1\), it then follows that \(S\) is contained in \(B\), and intersects
\(\rg{B}\).  Further, by the definition of the immediacy function, the Hamming
distance between \(T(x)_S\) and \(T(y)_S\) must be at least
$$
\delta\cdot\im(j t) > 2^{-\kappa}\im(j t) = {\abs{\lf{B}}}.
$$
Thus, since this Hamming distance is greater than the length of \(\lf{B}\), it
must be the case that \(T(x)_S\) and \(T(y)_S\) differ also on some index in
\(\rg{B}\), as we wanted to prove.  This proves that \(T\) is an
\((\alpha, \ell)\)-\lamCode{}, with \(\alpha = 2^{-(\kappa + 1)}\), as chosen above.

\paragraph*{Rate upper bounds with immediacy functions} Using the rate
upper bound for \((\alpha, \ell)\)-\lamCodes{} given by our main
result~(\cref{thm-lamcode-rate-upper-bound}), we thus get that if \(\Sigma\) is the
output alphabet of \(T\), it must be the case that (recall that
\(n = 2\cdot\im(\ell t(\delta))\))
\begin{equation}
  \label{eq:25}
  \lg \abs{\Sigma} \geq \frac{\ell}{2^{1+\kappa}} 
  = \frac{\im^{-1}(n/2)}{t(\delta) \cdot 2^{1+\kappa}} \geq \frac{\delta}{4t(\delta)}\im^{-1}(n/2).
\end{equation}
Put differently, the rate \(\rho\) of such a tree code satisfies
\begin{equation}
  \label{eq:26}
  \rho \leq \frac{4t(\delta)}{\delta \cdot \im^{-1}(n/2)}.
\end{equation}
(We remark here that since we work in the setting of constant distance
\(\delta\), we have not attempted to optimize the dependence on \(\delta\) in the above
computation.)

\noindent\textbf{Choosing \(t(\delta)\).} We now show how to choose \(t(\delta)\) so as to
satisfy the divisibility conditions of \cref{eq:24}, for the two immediacy
functions considered in the introduction.
\begin{description}
\item[Case 1: \(\im(k) = 2^k\).]  Put
  \(t(\delta) = 1 + \kappa = \floor{\lg(4/\delta)}\). A direct computation shows that the
  condition in \cref{eq:24} holds with this choice.  The rate upper bound in
  \cref{eq:26} becomes
  \begin{equation}
    \label{eq:27}
    \rho \leq \frac{4\lg(4/\delta)}{\delta \cdot \lg (n/2)}.
  \end{equation}

\item[Case 2: \(\im(k) = 2^{2^k}\).] Put
  \(t(\delta) \defeq \ceil{\lg (\kappa + 2)} = \ceil{\lg \floor{\lg(8/\delta)}}\).  Again, a
  direct computation shows that the condition in \cref{eq:24} holds with this
  choice.  The rate upper bound in \cref{eq:26} becomes
  \begin{equation}
    \label{eq:22}
    \rho \leq \frac{4\cdot \ceil{\lg \lg (8/\delta)}}{\delta\cdot\lg \lg (n/2)}.
  \end{equation}
\end{description}

\section{Rate upper-bound for \lamCodes{}}
\label{sec:rate-upper-bound}

In this section, we prove our main result: a rate upper bound for \lamCodes{}
(\cref{thm-lamcode-rate-upper-bound}).  The main ingredient is a simple
information theoretic tool, that we describe next.  As stated in the
introduction, given this tool and the definition of \lamCodes{}, the proof of
the main result becomes quite simple.

\subsection{Data processing inequality}
\label{sec:data-proc-ineq}
We will need the following simple consequence of (the usual proof of) the data
processing inequality.  Similar inequalities arise also in studies of the
Gács-Körner common information~(see, e.g., \cite{kamathNewDualGacsKorner2010}).
\begin{lemma}[A consequence of the data processing inequality]\label{lem-data-processing}
  Let \(A, B, C\) be discrete random variables such that
  \(H(A|B) = H(A|C) = 0\).  Then \(I(B:C) \geq H(A)\).  In particular,
  $$H(B) + H(C) \geq H(B, C) + H(A).
  $$
\end{lemma}
\begin{proof}
  Since \(A, B, C\) are discrete, the associated entropies and conditional
  entropies are non-negative.  Combining this with the non-negativity of mutual
  information and conditional mutual information, we thus have
  \begin{equation}
    \label{eq:8}
    0 \leq I(A:C|B) \leq H(A|B) = 0.
  \end{equation}
  We also have
  \(I(A:C) = H(A) - H(A|C) = H(A)\).  The chain rule for mutual information
  along with the non-negativity of the conditional mutual information then gives
  \begin{align*}
    H(A) &= I(A:C)  \leq I(A:C) + I(B:C|A) \\
         &= I(A,B : C) 
           =
           I(B:C) + I(A:C|B) \overset{\textup{\cref{eq:8}}}{=}  I(B:C),
  \end{align*}
  which proves that \(I(B:C) \geq H(A)\).  The final claim follows since
  \(H(B, C) = H(B) + H(C) - I(B:C)\).
\end{proof}

\subsection{Proof of the main result}
\label{sec:statement-proof-main}

  We now prove our main result (\cref{thm-lamcode-rate-upper-bound}).  We
  restate the result here for convenience of reference: please also see the
  remark at the end of \cref{sec-main-results} for further discussion.  
\lamUpperBound*

\begin{proof}
  To simplify notation, we shorten \(\ell(n)\) to \(\ell\) in the following.  Let
  \(P = (P_0, P_1, \dots, P_{\ell})\) be the \lamPart{} associated with the
  \lamCode{} \(c\).  Modify the code \(c\) so that it is systematic, in the
  sense that \(c(x)_k\) determines \(x_k\) for each \(1 \leq k \leq n\): this can be
  done by changing the alphabet to \(\Sigma \times \Sigma_{\mathrm{in}}\), and simply concatenating
  \(x_k\) to each \(c(x)_k\).  Note that even after the modification, the code
  \(c\) continues to be \alamCode{}, with the same associated \lamPart{} \(P\).
  However, the output alphabet of the code now is \(\Sigma'\) where
  \(\Sigma' = \Sigma \times \Sigma_{\mathrm{in}}\).

  Let \(X\) be uniformly distributed over \(\Sigma_{\mathrm{in}}^{n}\), and define
  \(Y = c(X)\).  We then have
  \(H(X) = H(Y) = n\cdot \lg \abs{\Sigma_{\mathrm{in}}}\), where we measure entropy in
  bits.  Now, consider any subset \(B\) that is an element of one of the
  partitions \(P_i\), for some \(i \geq 1\).  Since \(c\) is systematic, we then
  have \(H(X_{\lf{B}}|Y_{\lf{B}}) = H(X_{\lf{B}}|c(X)_{\lf{B}}) = 0\).  Further,
  by the neighborhood decoding property of the \lamCode{} \(c\), we also have
  \begin{equation}
    \label{eq:1}
    H(X_{\lf{B}}|Y_{\rg{B}}) = H(\phi_B(c(X)_{\rg{B}})|c(X)_{\rg{B}})  = 0.
  \end{equation}
  (Here, \(\phi_B\) is as in the definition of the neighborhood decoding property
  of \alamCode{}.)  Applying \cref{lem-data-processing}, we then have
  \begin{align}
    \label{eq:2}
    H(Y_B) = H(Y_{\lf{B}}, Y_{\rg{B}})
    &\leq H(Y_{\lf{B}}) + H(Y_{\rg{B}}) -
      H(X_{\lf{B}}) \\
    &= H(Y_{\lf{B}}) + H(Y_{\rg{B}}) - \abs{\lf{B}}\cdot\lg \abs{\Sigma_{\mathrm{in}}}\\
    &\leq H(Y_{\lf{B}}) + H(Y_{\rg{B}}) - \alpha(n)|B|\cdot \lg \abs{\Sigma_{\mathrm{in}}},\label{eq:7}
  \end{align}
  where the last inequality follows since \(\lf{B} \geq \alpha(n)|B|\), by the size
  property of \((\alpha(n), \ell(n))\)-\lamParts{}.  Note that \cref{eq:7} holds for
  every subset \(B\) that is part of some partition \(P_i\), \(i \geq 1\), of
  \(P\).  Define, for \(0 \leq i \leq \ell\)
  \begin{equation}
    \label{eq:14}
    T_i \defeq \sum_{B \in P_i}H(Y_B).
  \end{equation}
  We then have, for \(1 \leq i \leq \ell\),
  \begin{align}
    T_i &= \sum_{B \in P_i}H(Y_B)
          \overset{\textup{\cref{eq:7}}}{\leq}
          -\alpha(n)n\cdot \lg \abs{\Sigma_{\mathrm{in}}} + \sum_{B \in P_i}\inp{H(Y_{\lf{B}}) + H(Y_{\rg{B}})}\\
        &\leq -\alpha(n)n\cdot \lg \abs{\Sigma_{\mathrm{in}}} + \sum_{B'\in P_{i-1}}H(Y_{B'})
          = T_{i-1} - \alpha(n)n\cdot \lg \abs{\Sigma_{\mathrm{in}}}.    \label{eq:15}
  \end{align}
  where the second inequality follows from the laminar property of the
  \lamPart{} \(P\) (which ensures that each of \(\lf{B}\) and \(\rg{B}\) are
  disjoint unions of sets in the partition \(P_{i-1}\), for every set
  \(B \in P_i\)), along with the sub-additivity of entropy: for any two random
  variables \(Z_1\) and \(Z_2\), \(H(Z_1, Z_2) \leq H(Z_1) + H(Z_2)\).  By
  induction, \cref{eq:15} thus gives
  \begin{equation}
    \label{eq:16}
    T_{\ell} \leq T_0 - \alpha(n) n \ell \cdot \lg \abs{\Sigma_{\mathrm{in}}}.
  \end{equation}
  The subadditivity of entropy also gives \(T_{\ell} \geq H(Y) = n\cdot \lg \abs{\Sigma_{\mathrm{in}}}\), and
  \begin{equation}
    \label{eq:17}
    T_0 \leq \sum_{i=1}^nH(Y_i) \leq n \cdot (\lg \abs{\Sigma_{\mathrm{in}}} + \lg |\Sigma|), \text{ since each \(Y_i\) has
      support \(\Sigma' = \Sigma\times \Sigma_{\mathrm{in}}.\)}
  \end{equation}
  Substituting these in \cref{eq:16} gives
  \(\lg |\Sigma| \geq \alpha(n) \ell\cdot \lg \abs{\Sigma_{\mathrm{in}}}.\)
\end{proof}

The proof above is robust to certain small perturbations to the definition of
\alamCode{}. In particular, it can be easily adapted to establish rate upper
bounds for \(D\)-deficient \lamCodes{} as well, as we show below. We restate the
corresponding result (\cref{thm-deficient-lamcode-rate-upper-bound}) below for
ease of reference. \defLamUpperBound*

\begin{proof}
  The proof is a minor modification of the proof of
  \cref{thm-lamcode-rate-upper-bound}, but we include all the details for
  completeness.  Let \(P = (P_0, P_1, \dots, P_{\ell})\) be the \lamPart{}
  associated with the \(D\)-deficient \lamCode{} \(c\), and let the deficiency
  subsets \(S_i \subseteq P_{i}\) be as in the definition
  (\cref{def-deficient-lamcodes}).  Modify the code \(c\) so that it is
  systematic, in the sense that \(c(x)_k\) determines \(x_k\) for each
  \(1 \leq k \leq n\): this can be done by changing the alphabet to
  \(\Sigma \times \Sigma_{\mathrm{in}}\), and simply concatenating \(x_k\) to each
  \(c(x)_k\).  Note that even after the modification, the code \(c\) continues
  to be a \(D\)-deficient \lamCode{}, with the same associated \lamPart{} \(P\)
  and the same associated deficiency sets.  However, the output alphabet of the
  code now is \(\Sigma'\) where \(\Sigma' = \Sigma \times \Sigma_{\mathrm{in}}\).

  Let \(X\) be uniformly distributed over \(\Sigma_{\mathrm{in}}^{n}\), and define
  \(Y = c(X)\), so that \(H(X) = H(Y) = n\cdot \lg \abs{\Sigma_{\mathrm{in}}}\).  Now,
  consider any subset \(B\) that is an element of one of the partitions \(P_i\),
  for some \(i \geq 1\).  Since \(c\) is systematic, we then have
  \(H(X_{\lf{B}}|Y_{\lf{B}}) = H(X_{\lf{B}}|c(X)_{\lf{B}}) = 0\).  Further, in
  case \(B \not \in S_i\), then by the neighborhood decoding property of the
  deficient \lamCode{} \(c\), we also have
  \begin{equation}
    H(X_{\lf{B}}|Y_{\rg{B}}) = H(\phi_B(c(X)_{\rg{B}})|c(X)_{\rg{B}})  = 0.
  \end{equation}
  (Here, \(\phi_B\) is as in the definition of the neighborhood property of a
  deficient \lamCode{}.)  Applying \cref{lem-data-processing}, we then have
  \begin{align}
    H(Y_B) = H(Y_{\lf{B}}, Y_{\rg{B}})
    &\leq H(Y_{\lf{B}}) + H(Y_{\rg{B}}) -
      H(X_{\lf{B}}) \\
    &= H(Y_{\lf{B}}) + H(Y_{\rg{B}}) - \abs{\lf{B}}\cdot \lg \abs{\Sigma_{\mathrm{in}}}\\
    &\leq H(Y_{\lf{B}}) + H(Y_{\rg{B}}) - \alpha(n)|B|\cdot \lg \abs{\Sigma_{\mathrm{in}}},\label{eq:29}
  \end{align}
  where the last inequality follows since \(\lf{B} \geq \alpha(n)|B|\), by the size
  property of \((\alpha(n), \ell(n))\)-\lamParts{}.  Note that when
  \(i \geq 1\), \cref{eq:29} holds for every subset \(B \in P_i \setminus S_i\).  The only
  difference with the proof of \cref{thm-lamcode-rate-upper-bound} is that when
  \(B \in P_{i}\) is an element of the deficiency set \(S_i\), we however only
  have (by the sub-additivity of the entropy):
  \begin{equation}
    \label{eq:30}
    H(Y_B) = H(Y_{\lf{B}}, Y_{\rg{B}}) 
    \leq H(Y_{\lf{B}}) + H(Y_{\rg{B}}).
  \end{equation}
  Define, for \(0 \leq i \leq \ell\)
  \begin{equation}
    T_i \defeq \sum_{B \in P_i}H(Y_B).
  \end{equation}
  We then have, for \(1 \leq i \leq \ell\),
  \begin{align}
    T_i &= \sum_{B \in P_i}H(Y_B)
          \overset{\textup{\cref{eq:29,eq:30}}}{\leq}
          -\alpha(n)n\cdot \lg \abs{\Sigma_{\mathrm{in}}}
          + \alpha(n)\cdot \lg \abs{\Sigma_{\mathrm{in}}} \cdot \sum_{B \in S_i}\abs{B} \nonumber\\
        &\qquad + \sum_{B \in P_i}\inp{H(Y_{\lf{B}}) + H(Y_{\rg{B}})} \\
        &\leq -\alpha(n)n\cdot \lg \abs{\Sigma_{\mathrm{in}}} 
          + \alpha(n)\cdot \lg \abs{\Sigma_{\mathrm{in}}} \cdot \sum_{B \in S_i}\abs{B}
          + \sum_{B'\in P_{i-1}}H(Y_{B'}) \\
        &= T_{i-1} - \alpha(n)n\cdot \lg \abs{\Sigma_{\mathrm{in}}} 
          + \alpha(n)\cdot \lg \abs{\Sigma_{\mathrm{in}}} \sum_{B \in S_i}\abs{B}.    \label{eq:31}
  \end{align}
  where the second inequality follows from the laminar property of the
  \lamPart{} \(P\) (which ensures that each of \(\lf{B}\) and \(\rg{B}\) are
  disjoint unions of sets in the partition \(P_{i-1}\), for every set
  \(B \in P_i\)), along with the sub-additivity of entropy.  By induction,
  \cref{eq:31} thus gives
  \begin{equation}
    \label{eq:32}
    T_{\ell} \leq T_0 - \alpha(n) n \ell \cdot \lg \abs{\Sigma_{\mathrm{in}}}
    + \alpha(n)\cdot \lg \abs{\Sigma_{\mathrm{in}}} \cdot \sum_{i=1}^{\ell}\sum_{B \in S_i}\abs{B} 
    \leq T_0 - \alpha(n) \cdot \lg \abs{\Sigma_{\mathrm{in}}} \cdot \inp{n \ell  - D},
  \end{equation}
  where the last inequality uses the fact that the code is only \(D\)-deficient
  (\cref{eq:28} of \cref{def-deficient-lamcodes}).  The subadditivity of entropy
  yields \(T_{\ell} \geq H(Y) = n\cdot \lg \abs{\Sigma_{\mathrm{in}}}\), and also that
  \begin{equation}
    T_0 \leq \sum_{i=1}^nH(Y_i) \leq n \cdot (\lg \abs{\Sigma_{\mathrm{in}}} + \lg |\Sigma|), \text{ since each \(Y_i\) has
      support \(\Sigma' = \Sigma\times \Sigma_{\mathrm{in}}.\)}
  \end{equation}
  Substituting these in \cref{eq:32} gives
  \(\lg |\Sigma| \geq \alpha(n) \cdot \inp{\ell - \frac{D}{n}}\cdot \lg \abs{\Sigma_{\mathrm{in}}}.\)
\end{proof}

\section{Rate upper-bounds for immediacy properties of known constructions}
\label{sec:appl-rate-upper-bounds}

In this section, we describe the immediacy properties enjoyed by the tree code
constructions of \cite{CohenHS18}, \cite{EvansKS94} and \cite{GHKRZW15}.  The
first two are similar to---but not the same as---the informal notion of immediacy
discussed in the introduction.  However, as we show below, our general framework
of \lamCodes{} developed in \cref{sec:laminar-codes} still applies, and allows
us to conclude that the \(\Omega(1/\log \log n)\) and \(\Omega(1/\log n)\) rates achieved
by these constructions are the best possible (up to constant factors) given the
immediacy properties they achieve. The third construction, due to
\cite{GHKRZW15}, tackles a different trade-off: what is the best (even if
vanishing, as \(n\) increases) distance one can achieve for a tree code if one
imposes the condition that the rate has to be constant.  For this construction
we show that the distance \(\Omega(1/\log n)\) distance it achieves with a constant
rate is tight up to a \(\log \log n\) factor for the immediacy guarantee that
this construction provides.

\subsection{The CHS construction}
We start with the construction of \cite{CohenHS18} (which for brevity we shall
call the CHS construction).  To describe the immediacy properties of this
construction, we first import some of the notation set up in \cite{CohenHS18}.
In the following, references to theorems, pages etc.~in \cite{CohenHS18} refer
to the ECCC version of \cite{CohenHS18}. Define a sequence of length scales
(following the proof of Theorem 1.1 in \cite[p.~16]{CohenHS18}) given by
$\ell_1 \defeq 2^{20}$ and $\ell_{i+1} \defeq \ell_i^2/2^{10}$ for $i \geq 1$; thus
\(\ell_i = 2^{10}\cdot32^{2^{i}}\) for every positive integer \(i\).  We consider a
transmission length \(n\) of the form \(\ell_{m+1}\) for some positive integer
\(m\).  Thus, $n$ is divisible by each $\ell_i$, $1 \leq i \leq m$, and no divisibility
issues arise in the following description.  We also note for future use that for
\(i \geq 1\)
  \begin{equation}
    \label{eq-l-estimate}
    16 \sqrt{\ell_i} \leq \frac{3\ell_i}{8}.
  \end{equation}

  For each length scale \(\ell_i\), \(1 \leq i \leq m + 1\) as above, we consider the
  partition \(P_{i-1}\) of $[n]$ into $2n/\ell_i$ consecutive disjoint intervals of
  length $\ell_i/2$ each.  In agreement with the terminology in \cite{CohenHS18},
  we refer to the elements of the \(P_i\) as \emph{blocks}.  For a block
  \(B \in P_i\) for \(i \geq 1\), we denote by \(\lf{B}\) the set consisting of the
  first \(|B|/4\) positions in \(B\), and by \(\rg{B}\) the set consisting of
  the last \(3|B|/4\) positions in \(B\).  With these definitions,
  \(P \defeq (P_0, P_1, P_2, \dots, P_m)\) is a \((1/4, m)\)-\lamPart{} of
  \([n]\).  {In particular, for any \(B \in P_i\), \(i \geq 1\), \(\lf{B}\) and
    \(\rg{B}\) are disjoint unions of blocks in \(P_{i-1}\), because
    \(\ell_i/2\) divides \(\ell_{i+1}/8\)}.

\paragraph*{The CHS immediacy condition} We can now describe the immediacy-like
property satisfied by the CHS tree code construction.  Denote their code by
\(T\).  Let \(x\) and \(y\) be two strings in \(\inb{0, 1}^n\) which differ at a
position \(s' \in [n]\).  Then, for any \(i\) such that \(2 \leq i \leq m + 1\) such
that \(s'\) is \emph{not} in the rightmost block in \(P_{i-1}\) (i.e., such that
\(s' \leq n - \ell_{i}/2\)), the following is true.  Let \(s(i)\) denote the
\emph{leftmost} position in the unique block \(B \in P_{i-1}\) containing \(s'\)
such that \(x_{s(i)} \neq y_{s(i)}\).  Then, for any $d$ satisfying
\begin{equation}
  \label{eq-d-bounds}
  \ell_{i-1}/2 =  16\sqrt{\ell_i} \leq d \leq \ell_i/2,
\end{equation}
it holds that
\begin{equation}
  \label{eq-error-condition}
  \ham\inp{T(x)_{I(d)}, T(y)_{I(d)}} \geq d/3,
\end{equation}
where $I(d)$ denotes the interval
\begin{equation}
  \label{eq-interval-d}
  I(d) \defeq \insq{s(i),  s(i) + d}.
\end{equation}
The above condition is implicit in the proof given in \cite{CohenHS18}: in
particular, it follows by substituting \(s(i)\) above in the role played by the
``split'' \(s\) in the proof of Claim 6.4 of \cite{CohenHS18}.

We then have the following consequence of the above immediacy condition.

\begin{lemma}
\label{lem-determine}
Let \(m\) be a positive integer.  Let the sequence \(\ell_1, \ell_2, \dots\), the
positive integer \(n = \ell_{m+1}\), and the \((1/4, m)\)-\lamPart{}
\(P = (P_0, P_1, P_2, \dots, P_m) \) of \(n\) be as defined above.  Let
\(T: \inb{0, 1}^n \rightarrow \Sigma^n\) be a code satisfying the CHS immediacy condition
described above.
Consider any block \(B\) in a partition \(P_{i-1} \in P\),
\(2 \leq i \leq m + 1\), so that \(B\) is {not} the rightmost block in
\(P_{i-1}\).  Then, there exists a function \(\phi_B\) such that for any
\(x \in \inb{0, 1}^n\), \(\phi_{B}(c(x)_{\rg{B}}) = x_{\lf{B}}\).  In other words,
the prefix \(x_{\lf{B}}\) of \(x_{B}\) is uniquely determined given the suffix
\(c(x)_{\rg{B}}\) of \(c(x)_B\).
\end{lemma}
\begin{proof}
  It is sufficient to show that if \(x, y \in \inb{0, 1}^n\) are such that
  \(x_{\lf{B}} \neq y_{\lf{B}}\) then it must be the case that
  \(c(x)_{\rg{B}} \neq c(y)_{\rg{B}}\).  Consider therefore
  \(x, y \in \inb{0, 1}^n\) such that \(x_{\lf{B}} \neq y_{\lf{B}}\).  Thus, there
  must be a position \(s'\) in \(\lf{B}\) such that \(x_{s'} \neq y_{s'}\).  Let
  \(s(i)\) be the leftmost such position in \(\lf{B}\).  Choose \(d\) so that
  the interval
  $$
  I(d) = [s(i), s(i) + d] = B \cap [s, n],
  $$
  i.e., so that
  \(I(d)\) is the suffix of \(B\) starting at \(s(i)\).  Since
  \(\abs{B} = \ell_i/2\), this \(d\) is of the form \(\ell_i/2 - k(i)\), where
  \(k(i) \geq 1\) is the relative index, counting starting with one from the left,
  of \(s(i)\) within \(B\).  Since \(s(i) \in \lf{B}\) and \(\lf{B}\) consists of
  the leftmost \(\abs{B}/4 = \ell_i/8\) positions in \(B\), it is also the case
  that \(k(i) \leq \ell_{i}/8\).  Given \cref{eq-l-estimate}, \(d\) therefore
  satisfies the lower bound required in \cref{eq-d-bounds} since
  \(k(i) \leq \ell_i/8\) and \(i \geq 2\).  It also satisfies the upper bound required in
  \cref{eq-d-bounds} by construction.  Since \(B\) is, by hypothesis, not the
  rightmost block in \(P_{i-1}\), \cref{eq-error-condition} then yields that
  \begin{equation}
    \label{eq:3}
    \ham\inp{c(x)_{I(d)}, c(y)_{I(d)}} \geq d/3 = \frac{\ell_i}{8} + \frac{\ell_i/8 - k(i)}{3}.
  \end{equation}
  Suppose, if possible, that \(c(x)_{\rg{B}} = c(y)_{\rg{B}}\), so that
  \(\ham\inp{c(x)_{\rg{B}}, c(y)_{\rg{B}}} = 0.\) Let \(S\) denote the suffix of
  \(\lf{B}\) starting from the position \(s(i)\), so that
  \(I(d) = S \sqcup \rg{B}\).  We then have
  \begin{align}
    \label{eq:6}
    \ham\inp{c(x)_{I(d)}, c(y)_{I(d)}} 
    &= \ham\inp{c(x)_{S}, c(y)_{S}} +
      \ham\inp{c(x)_{\rg{B}}, c(y)_{\rg{B}}} \\
    &= \ham\inp{c(x)_{S}, c(y)_{S}}
    \leq |S|
    = \frac{\ell_i}{8} - k(i) + 1.
  \end{align}
  However, this contradicts \cref{eq:3} since \(\ell_i \geq 32\) and
  \(k(i) \geq 1\).  It must therefore be the case that
  \(c(x)_{\rg{B}} \neq c(y)_{\rg{B}}\).
\end{proof}

It is an easy consequence of the above lemma that any binary tree code
satisfying the CHS immediacy condition is a deficient \lamCode{}, with a small
deficiency parameter.

\begin{proposition}\label{prop-chs-immediate}
  Let \(m\) be a positive integer.  Let the sequence \(\ell_1, \ell_2, \dots\), the
  positive integer \(n = \ell_{m+1}\), and the \((1/4, m)\)-\lamPart{}
  \(P = (P_0, P_1, P_2, \dots, P_m) \) of \(n\) be as defined above.  Let
  \(T: \inb{0, 1}^n \rightarrow \Sigma^n\) be a code satisfying the CHS immediacy condition.
  Then \(T\) is an \(n\)-deficient \((1/4, m)\)-\lamCode{}.
\end{proposition}

\begin{proof}
  For each \(1 \leq i \leq m\) define \(S_i\) to consist only of the rightmost block
  in \(P_i\).  Then, it follows directly from \cref{lem-determine} that \(T\) is
  a \(D\)-deficient \((1/4,m)\)-\lamCode{} with deficiency sets
  \(S_1, S_2, \dots, S_{m}\), provided \(D\) satisfies
  \begin{equation}
    \label{eq:23}
    D \geq \sum_{i=1}^m\sum_{B \in S_i}|B|.
  \end{equation}
  Now, the latter quantity can be bounded from above as
  \begin{equation}
    \frac{n}{2}\sum_{i=1}^m\frac{\ell_{i+1}}{\ell_{m+1}} \leq n,
  \end{equation}
  using the growth condition on the \(\ell_i\).  Thus the choice \(D=n\) works.
\end{proof}

Applying \cref{thm-deficient-lamcode-rate-upper-bound}, we thus get that the
output alphabet size of any binary tree code satisfying the CHS immediacy
condition with transmission length \(n\) as above satisfies
\begin{equation}
  \label{eq:33}
  {\lg \abs{\Sigma}} \geq \frac{m-1}{4} = \Omega(\lg \lg n),
\end{equation}
since \(n = 2^{10}\cdot32^{2^{m+1}}.\) Thus, the rate \(\rho\) of such a code must
satisfy \(\rho = O(1/\lg \lg n)\).

\subsubsection{Subexponential time constructions of good tree codes}
\label{sec-subexp-time-constr}

Braverman~\cite{Braverman12} presented a generic procedure to explicitly
construct good tree codes in sub-exponential time.  More formally, given
\(\epsilon \in (0, 1)\), Braverman's procedure runs in time
\(\exp(O(n^{\epsilon}))\) and outputs a depth \(n\) tree code that has constant rate
(depending upon \(\epsilon\)) and fixed constant distance.  At a high level,
Braverman's construction works by starting with a brute force
\(\exp(O(n^{\epsilon}))\)-time construction of a tree code of depth
\(n^{\epsilon}\), and then using a sequence of (efficiently implementable) tilings of
this intermediate tree code, augmented with error-correcting codes, in order to
handle error-propagation at different length scales.  However, as described in
\cite{Braverman12}, this construction achieves a final alphabet size that is at
least exponential in \(1/\epsilon\), so that the rate achieved is no better than
\(O(1/\epsilon)\).

We first review here a refinement of the above construction, mentioned in
passing in \cite[remarks following Theorem 1.1]{CohenHS18}, that constructs a
constant distance tree code with rate \(\Theta(1/\log(1/\epsilon))\) and depth
\(n\) in time \(\exp(O(n^{\epsilon}))\).  We then describe the immediacy properties of
this construction, and show that no code with these immediacy properties can
have rate better than \(\Omega(1/\log(1/\epsilon))\).  The arguments are very similar to
those carried out above for the immediacy properties of the CHS construction.

\paragraph*{Review of a subexponential-time constructible good tree code} We use
the notation from \cite{CohenHS18} imported above.  Let \(k\) be a fixed
positive integer parameter (corresponding roughly to \(\log(1/\epsilon)\)).  Consider
an \(n\) of the form \(\ell_{m+1} = 2^{10}\cdot32^{2^{m + 1}}\) as above. The
polynomial time (in \(n\)) construction of a tree code of depth \(n\) given in
the proof of Theorem 1.1 of \cite{CohenHS18} consists principally of a
concatenation of constant rate intermediate codes \(C_i\), for each
\(1 \leq i \leq m + 1\), where the \(i\)th code corresponds to length scale
\(\ell_i=2^{10}\cdot 32^{2^i}\).  In consequence, the resulting final code has rate
\(\Theta(1/\lg \lg n)\). 

To get a better rate at the cost of an increase in the runtime of the
construction, start by replacing the above construction by a concatenation of
only the codes \(C_{m-k+1}, C_{m-k+2}, \dots, C_{m+1}\).  Following the proof of
Theorem 1.1 of \cite{CohenHS18}, one then gets that this resulting intermediate
code is, in the terminology of \cite{CohenHS18}, a \emph{lagged tree code}
(i.e., where error propagation from an error starts after a certain lag) with
lag \(16\sqrt{\ell_{m-k+1}} = \ell_{m-k}/2\).  Applying a variant of the idea from
\cite{Braverman12}, one can then remove the lag as follows.  First construct, in
time \(\exp(O(\ell_{m-k}))\), a constant rate and constant distance tree code of
depth \(\ell_{m-k+1}\). Then, concatenate to the above intermediate construction
two tilings of this good tree code staggered by a shift of \(\ell_{m-k+1}/2\)
(following, e.g., Lemma 5 of \cite{Braverman12}).  The resulting construction is
a tree code of constant distance and rate \(\Theta(1)/k\), obtained in time
\(\exp(O(\ell_{m-k+1})) = \exp(O(n^{1/2^k}))\).  Reparameterizing with
\(\epsilon = 1/2^{k}\), we see that the rate is \(\Theta(1)/\lg(1/\epsilon)\) and the runtime of
the construction is \(\exp(O(n^{\epsilon}))\).

\paragraph*{Immediacy properties} We now consider the immediacy properties of the
the above construction.  The arguments before \cref{lem-determine}, when applied
to above variant of the CHS construction, show that the immediacy conditions
described in \cref{eq-d-bounds,eq-error-condition,eq-interval-d} hold for this
version for all \(i\) satisfying \(m - k + 2 \leq i \leq m + 1\). (For the original
CHS construction, these conditions held for all \(i\) such that
\(2 \leq i \leq m + 1\).)

\paragraph*{Rate upper bound} We now show that any tree code satisfying these
immediacy properties cannot have rate better better than \(O(1/k)\).  To do
this, we repeat the proof of \cref{lem-determine,prop-chs-immediate} with the
above modified version of the CHS immediacy condition to obtain that any such
code must be an \(n\)-deficient \((1/4, k)\)-\lamCode{} with respect to the
\((1/4, k)\)-\lamPart{} \((P_{m-k}, P_{m-k+1}, \dots, P_{m})\) (where the
\(P_i\) are as defined in the above description of the original CHS
construction, following \cref{eq-l-estimate}).  Applying \cref{thm-deficient-lamcode-rate-upper-bound}, we
therefore get that the output alphabet size of any depth \(n\) binary tree code satisfying
the above variant of the CHS immediacy condition satisfies
\begin{equation}
  \label{eq:9}
  {\lg \abs{\Sigma}} \geq \frac{k-1}{4}.
\end{equation}
Thus, the rate \(\rho\) of such a code must satisfy \(\rho = O(1/k)\).

\subsection{The EKS construction}
\label{sec:eks-construction}
\newcommand{\ECC}{\ensuremath{\mathrm{ECC}}}

We now turn to the construction of \cite{EvansKS94} (which for brevity we shall call the EKS construction.)   A version of this construction can be
described as follows.  Let
\(\ECC_{\ell}: \inb{0, 1}^{\ell} \rightarrow \inb{\inb{0, 1}^{b}}^{\ell}\) denote an
error-correcting block code with distance \(\delta \in (0, 1)\).  (It is known that
there is a choice of \(b = b(\delta)\), independent of the block-length
\(\ell\), so that such a code exists for each positive integer \(\ell\), and is such
that \(\ECC_{\ell}\) can be computed in time polynomial in \(\ell\): see, e.g.,
\cite[Lemma 3.2 in the ECCC version]{CohenHS18}.)  {Note that $\ECC_1$
  can be chosen to be the repetition code of length $b$.}
For simplicity, we describe
the construction of the EKS tree code \(T\) when the depth \(n = 2^k\) is a
power of \(2\).  For \(0 \leq i \leq k\), let \(P_i\) be the partition of
\([n]\) consisting of \(n/2^i\) consecutive disjoint blocks of length \(2^i\)
each, {and we write $P_{i,j}=\{(j-1)2^i+1,\ldots,j\cdot2^i\}$. For a message \(x \in \inb{0, 1}^n\) we let $x_{P_{i,j}} \in \inb{0,1}^{2^i}$ denote the restriction of $x$ to the interval $P_{i,j}$. 
} To describe the encoding \(T(x)\) of a message \(x \in \inb{0, 1}^n\), we
form the following \((k + 1) \times n\) table {with entries in \(\inb{0, 1}^b\)}:

\begin{enumerate}
\item {The first row consists of the concatenation of $\ECC_1(x_j)$ for $j=1,\ldots,n$.}
\item {The \(i\)th row, for \(2 \leq i \leq k + 1\), consists of
    $(0^b)^{2^{i-2}}$ followed by the concatenation of
    $\ECC_{2^{i-2}}(x_{P_{i-2,j}})$ for $j=1,\ldots,n \cdot 2^{2-i}-1$. Informally, the
    code blocks in the \(i\)th row are ``right-shifted'' by
    \(2^{i-2}\).} \label{item-shift}
\end{enumerate}

The encoding \(T(x)\) is defined by letting its \(j\)th character, for
\(1 \leq j \leq n\), be the \(j\)th column (from the left) of this table.  Note that
because of the ``right shift'' operation in \cref{item-shift}, this construction
satisfies the ``online'' requirement for tree codes.  The output alphabet of
\(T\) is {$\{0,1\}^{b(k+1)}$.}

\paragraph*{The EKS immediacy condition}  From the construction, it is easy to
see that the EKS construction satisfies the following version of immediacy.
Denote the code by \(T\), and let \(x, y \in \inb{0, 1}^n\) be two strings that
differ at a position \(s' \in [n]\).  Let \(0 \leq \ell < k\) be any integer such that
\(s' \leq n - 2^\ell = 2^k - 2^{\ell}\).  Define
\(s = s(\ell) \defeq \ceil{s'/2^\ell}\cdot2^\ell\) to be the smallest multiple of
\(2^\ell\) that is at least as large as \(s'\).  Then
\begin{equation}
  \label{eq:4}
  \ham(T(x)_{(s, s + 2^\ell]}, T(y)_{(s, s + 2^\ell]}) \geq \delta \cdot 2^\ell > 0.
\end{equation}

We now show that any code satisfying the immediacy condition described by
\cref{eq:4} must be an \((\Omega(1), \Omega(\log n))\)-\lamCode.  To see this, we consider
the partitions \(P_i\), \(0 \leq i \leq k\), of \([n] = [2^k]\) considered above, and
convert each of them into \atagPart{} by defining, for each \(B \in P_i\) with
\(i \geq 1\), \(\lf{B}\) to be the leftmost \(2^{{i-1}}\) entries in \(B\) and
\(\rg{B}\) to be the rightmost \(2^{i-1}\) entries in \(B\).  With this
construction, it is immediate that \((P_0, P_1, P_2, \dots, P_k)\) is a
\((1/2, k)\)-\lamPart{}.  Since the code \(T\) satisfies the immediacy condition
in \cref{eq:4}, it is also immediate that it satisfies the neighborhood decoding
property for each \(B \in P_i\) with \(i \geq 1\).  We thus see that any tree code
\(T\) with depth \(n = 2^k\) and satisfying the immediacy condition given by
\cref{eq:4} must be a \((1/2, k)\)-\lamCode{}.  Applying
\cref{thm-lamcode-rate-upper-bound}, we thus see that the output alphabet
\(\Sigma\) of such a code must satisfy
\begin{equation}
  \label{eq:5}
  \abs{\Sigma} \geq \frac{k}{2} = \Omega(\lg n),
\end{equation}
since \(n = 2^k\).  Thus the rate \(\rho\) of such a code must be \(O(1/\lg n)\).

\NewDocumentCommand{\ghk}{}{GHKRZW}

\subsection{The \texorpdfstring{\ghk{}}{GHK+} construction}
\label{sec:ghk-construction}
We now turn to a tree code construction by Gelles, Haeupler, Kol, Ron-Zewi and
Wigderson~\cite{GHKRZW15} (which we shall refer to as the \ghk{}
construction).  While related to the EKS construction discussed in
\cref{sec:eks-construction}, this construction is qualitatively different from
the previous constructions considered in this paper in that it insists upon a
\emph{constant}, non-vanishing rate even if at the cost of a vanishing
\emph{distance} parameter.  We show here that for the immediacy guarantee the
\ghk{} construction achieves, and given its requirement of a constant rate, the
\(\Omega(1/\log n)\) vanishing distance it achieves is tight up to a \(\log \log n\)
factor.

The \ghk{} construction (described in \cite[Section 5.1]{GHKRZW15}) demonstrates
the following: there exists a positive integer \(k_0\), such that for every
\(\epsilon \in (0, 1)\) there is a tree code \(T\), with transmission length
\(n\), input alphabet
\(\Sigma_{\mathrm{in}} \defeq \inb{0, 1}^{(\lg n)/\epsilon}\), output alphabet
\(\Sigma \defeq \inb{0, 1}^{(1 + 1/\epsilon)\cdot\lg n}\) (so that its rate is
\(1/(1+\epsilon)\)) and distance parameter \(\delta\) is
\(\frac{\epsilon}{32k_0\cdot \lg n }\).  Here, it is convenient to assume (as we do now)
that \(\lg n, k_0\) and \(1/\epsilon\) are all powers of two, and that
\(k_0 \geq 16\).  (The notation and conventions here are essentially exactly those
adopted in \cite[Section 5.1]{GHKRZW15}).

\paragraph*{The \ghk{} immediacy condition}  The distance proof for the \ghk{}
construction, given in Section 5.2 of \cite{GHKRZW15}, in fact establishes the
following immediacy property.  Let \(k_{0}, \epsilon\) be parameters not growing with
the transmission length \(n\) (as above) and let
\(m \defeq (k_0/\epsilon) \cdot \lg n\) be as defined in \cite{GHKRZW15}; note that
\(m\) is also a power of two by the assumptions adopted above.  Let \(t\) be any
integer satisfying \(\lg m \leq t \leq \lg n - 1\).  Suppose that
\(x, y \in \Sigma_{\mathrm{in}}\) differ at a position \(i \leq n - 2^{t}\), and Let
\(i_0 = i_0(t) \defeq \floor{(i - 1)2^{-t}}\cdot2^t\) be the largest multiple of
\(2^t\) that is smaller than \(i\), as defined in \cite[Section 5.2]{GHKRZW15}.
Then,
\begin{equation}
  \label{eq-ghk-immediacy}
  \ham\inp{
    T(x)_{(i_{0}, i_0 + 2^{t+1}]},
    T(y)_{(i_{0}, i_0 + 2^{t+1}]}
  } \geq \delta\cdot2^{{t+1}}.
\end{equation}
 
To prove this, one repeats the proof in Section 5.2 of \cite{GHKRZW15}, ignoring
the parameter \(j\) in that proof, and also ignoring the condition, stipulated
in the proof, that \(i\) is the \emph{first} position on which \(x\) and \(y\)
differ.  The last paragraph on p.~1930 in \cite{GHKRZW15}, combined with the
first display on p.~1931 (neither of which use the latter condition that \(i\)
is the first position on which \(x\) and \(y\) differ), then yields (a slightly
stronger version of) the above condition, after we note that \(\delta'\) in that
proof is at least \(8\delta\), where \(\delta\), as above, is the distance parameter
achieved by the \ghk{} construction.\footnote{For the convenience of the reader,
  we note here that there is a minor typographical error on p.~1930 of
  \cite{GHKRZW15}, where the ``\(\max\)'' in the first display in Section 5.2
  should actually be ``\(\min\)''.}

We now show that any constant rate tree code satisfying the immediacy condition
specified by \cref{eq-ghk-immediacy} cannot have distance much better than the
\ghk{} construction.  The proof is similar to the case of an exponential
immediacy function analyzed in \cref{sec:immed-funct-lamc}.

As in that proof, our goal is to show that such a code is \alamCode{} with
appropriate parameters.  We begin by defining the corresponding \lamPart{}.  Let
\(\delta \in (0, 1)\) be the distance parameter (possibly dependent upon \(n\)) of such
a code.  Let \(\kappa \defeq \floor{\lg(2/\delta)}\) be a positive integer so that
\(2^{-\kappa} < \delta \leq 2^{-\kappa + 1}\).  Set
\begin{equation}
  \label{eq:12}
  \ell \defeq 1 + \floor{(\lg{(n/(2m))})/\kappa}.
\end{equation}
Let \(P_0\) be the partition of \([n]\) into \(n\) singletons, and for
\(1 \leq i \leq \ell\), define \(P_{i}\) to the partition of \([n]\) into blocks of
consecutive integers, each of length \(n/2^{\kappa \cdot (\ell - i)}\).  Thus, for
\(1 \leq i \leq i + 1\), any set \(B\) in \(P_i\) is an interval of the form
\((b\cdot2^{t+1}, (b + 1)\cdot2^{t+1}]\), where \(b\) is an integer and
\(t = \lg n - {\kappa \cdot (\ell - i) - 1}\).  For \(1 \leq i \leq \ell\), convert
\(P_i\) into a \tagPart{} by defining, for each \(B \in P_i\), \(\lf{B}\) to be
the set of the smallest \(2^{-\kappa}\cdot\abs{B}\) elements in \(B\) (so that
\(\rg{B} = B \setminus \lf{B}\)).  It is easy to verify that the above construction
gives a \((2^{-\kappa}, \ell)\)-\lamPart{} of \([n]\).

We now show that any code satisfying the immediacy condition of
\cref{eq-ghk-immediacy} must be a \((2^{-\kappa}, \ell)\)-\lamCode{}.  We use the
\lamPart{} defined above.  Fix \(1 \leq j \leq \ell\) and \(B \in P_j\).  Let
\(t = \lg n - {\kappa \cdot (\ell - j) - 1}\) be as above (note that \(t\) satisfies
\(\lg m \leq t \le \lg n - 1\)), so that
\(B = (b\cdot2^{t+1}, (b + 1)\cdot2^{t+1}]\) for some non-negative integer
\(b\). Consider any \(x, y \in \Sigma_{\mathrm{in}}^{n}\) that differ on
\(\lf{B} = (b\cdot2^{t+1}, (b + 2^{-\kappa})\cdot2^{t+1}]\).  Thus, there exists
\(i \in \lf{B} \subseteq (b\cdot2^{t+1}, (b + 1/2)\cdot2^{t+1}]\) such that
\(x_i \neq y_i\).  The immediacy condition of \cref{eq-ghk-immediacy} then gives
that
\begin{equation}
  \label{eq:10}
  \ham\inp{T(x)_{B}, T(y)_B} \geq \delta \cdot 2^{{t+1}} > 2^{-\kappa} \cdot 2^{{t+1}} = \abs{\lf{B}}.
\end{equation}
Here the second inequality uses \(\delta > 2^{-\kappa}\).  Thus, we see that the Hamming
distance between \(T(x)_B\) and \(T(y)_B\) is strictly more than the size of
\(\lf{B}\) whenever \(x_{\lf{B}}\) and \(y_{\lf{B}}\) differ.  This implies that
\(T(x)_{\rg{B}}\) and \(T(y)_{\rg{B}}\) must differ whenever \(x_{\lf{B}}\) and
\(y_{\lf{B}}\) differ.  This establishes that \(T\) is a
\((2^{-\kappa}, \ell)\)-\lamCode{}.

We now apply \cref{thm-lamcode-rate-upper-bound} to see that a tree code
\(T: \Sigma_{\mathrm{in}}^n \rightarrow \Sigma^{n}\) satisfying the immediacy condition of
\cref{eq-ghk-immediacy} must therefore satisfy
  \begin{equation}
    \label{eq:11}
    \frac{\lg \abs{\Sigma}}{\lg \abs{\Sigma_{\mathrm{in}}}} 
    \geq 2^{-\kappa} \cdot \ell 
    \overset{\textup{\cref{eq:12}}}{\ge}
      \frac{\delta \cdot \lg \frac{n}{2m}}{2\cdot \lg(2/\delta)}. 
  \end{equation}
  Recall that \(m = (k_0/\epsilon) \cdot \lg n\), where \(k_{0}, \epsilon\) do not grow with
  \(n\).  Thus, \cref{eq:11} shows that any constant rate tree code satisfying
  the immediacy condition \cref{eq-ghk-immediacy} of the \ghk{} construction
  must have a distance parameter \(\delta\) satisfying
  \begin{equation}
    \label{eq:13}
    \frac{\delta}{\log(1/\delta)} = O(1/\log n).
  \end{equation}
  In particular, a code with the immediacy guarantee that the \ghk{}
  construction achieves cannot achieve a distance parameter \(\delta\) that is
  \(\omega\inp{\frac{\log \log n}{\log n}}\).

\bibliography{refs}

\end{document}